\documentclass[conference]{IEEEtran}
       \usepackage{amsmath}
\usepackage{amsmath,amsthm}
\usepackage{blindtext, graphicx}

\usepackage{mathptmx}

\usepackage{thmtools,thm-restate}
\usepackage{mathtools}

\usepackage{tikz}
\usepackage{tabu}
\usepackage{graphicx} 
\usepackage{caption} 

\usetikzlibrary{arrows}
\theoremstyle{definition}

\theoremstyle{plain}

\usepackage{algpseudocode}
\usepackage{algorithm}
\usepackage{array}
\newcolumntype{?}{!{\vrule width 2pt}}
\usepackage{booktabs}
\algnewcommand\algorithmicforeach{\textbf{for each}}
\algdef{S}[FOR]{ForEach}[1]{\algorithmicforeach\ #1\ \algorithmicdo}

\usepackage{makeidx}  % allows for indexgeneration
\begin{document}
%
%\frontmatter          % for the preliminaries
%
\pagestyle{headings}  % switches on printing of running heads
%

%
%\mainmatter              % start of the contributions
%
\title{Solving Graph Isomorphism Problem for a Special case}
%
%\titlerunning{GI for a special case}  % abbreviated title (for running head)
% \author{Vaibhav Amit Patel }
%
%\authorrunning{Vaibhav Amit Patel et al.} % abbreviated author list (for running head)
%\tocauthor{Vaibhav Amit Patel}
%
\author{
    \IEEEauthorblockN{Vaibhav Amit Patel}
    \IEEEauthorblockA{Dhirubhai Ambani Institute of Information and Communication Technology Gandhinagar, Gujarat, 382007 India,\\ vaibhav290797@gmail.com}
}
% }
% \institute{Dhirubhai Ambani Institute of Information and Communication Technology Gandhinagar, Gujarat, 382007 India,\\
% \email{vaibhav290797@gmail.com}}
\newcommand\tab[1][0.6cm]{\hspace*{#1}}
\maketitle              % typeset the title of the contribution

\begin{abstract}
Graph isomorphism is an important computer science problem. The problem for the general case is unknown to be in polynomial time. The bese algorithm for the general case works in quasi-polynomial time \cite{babai2016graph}. 
The solutions in polynomial time for some special type of classes are known. In this work, we have worked with a special type of graphs. We have proposed a method to represent these graphs and finding isomorphism between these graphs. The method uses a modified version of the degree list of a graph and neighbourhood degree list \cite{barrus2015neighborhood}. These special type of graphs have a property that neighbourhood degree list of any two immediate neighbours is different for every vertex.The representation becomes invariant to the order in which the node was selected for giving the representation making the isomorphism problem trivial for this case. The algorithm works in $O(n^4)$ time, where n is the number of vertices present in the graph. The proposed algorithm runs faster than quasi-polynomial time for the graphs used in the study.  
\IEEEkeywords{graph theory, graph isomorphism problem }
\end{abstract}
\section{Introduction}
Graph is a popular data structure and it can be used in many complex real world applications, such as social networks, networking. For example, two people on a social networking site a and b can be represented by a graph consisting nodes $v_a$ and $v_b$. Now, if they are friends, then their relationship can be represented by an edge e($v_a$,$v_b$) between them. Graph theory is very important in many other application. There are many open problems in the graph theory. Description of graph isomorphism problem is easy, but difficult to solve. Graph isomorphism: It is a bijection on vertex set of graph G and H that preserves edges. Graph isomorphism problem is a special case of subgraph isomorphism problem which is in NP-complete complexity class.
\par Checking whether two Graphs are isomorphic or not is an old and interesting computational problem. The simplest,but inefficient approach for checking two graphs G and H, each with vertices v is making all the permutation of the nodes and see if it is edge preserving bijection or not. Obviously, this solution takes $O(v^2v!)$ time, where $v$ is the number of vertices of the graph. Even for smaller values of $v$ the problem becomes intractable because of this enormous time complexity.

\par 
Babai \cite{babai2016graph} has shown the general case algorithm to be solved in quasi-polynomial time. While GI problem for a general case is a hard problem, many works have been done on different special types of graphs. These special types are important in graph theory and polynomial time solution is known for them.
Kelly et. al. \cite{kelly1957congruence} and Aho et. al \cite{aho1974design} have worked on isomorphism on trees and they have found a poolynomial time algorithm for this case. Hopcraft et. al. found an algorithm for solving GI on planner graphs \cite{hopcroft1974linear}. In fact, this problem is in log space.
Lueker et. al. have worked on interval grpahs \cite{lueker1979linear}. Colbourn et. al. found an algorithm to solve GI on permutation graphs \cite{colbourn1981testing}.
A lot of work has been done on the solving GI for bounded parameters of a graph.
Bodlaender et. al. have worked on graphs of bounded treewidth \cite{bodlaender1990polynomial}. 
Filotti et. al. have worked on graphs of bounded genus \cite{filotti1980polynomial}. 
L. Babai et. al. have worked on GI on Graphs with bounded eigenvalue multiplicity \cite{babai1982isomorphism}.
Miller et. al. have worked on k-Contractible GI problem \cite{miller1983isomorphism}.
Luks et. al. have worked on color-preserving isomorphism of colored graphs with bounded color multiplicity \cite{luks1986parallel}. Zager et. al. have worked on this problem using similarity functions \cite{zager2008graph}. 
In this work, first we will define the graphs that are used in this study. After that, we will discuss the algorithm for solving isomorphism for this type of graphs in polynomial time. A worked out example is also given along with the algorithm. 

\section{Algorithm}\label{sec:algorithm}

This is a simple example of GI problem. We have given two graphs $G_1$ and $G_2$. The problem is that we have to find whether these two graphs are isomorphic or not and if they are isomorphic then give the isomorphism. In practice, the the input is the adjacency matrices of these two graphs. Two isomorphic graphs have equal number of nodes, number of edges, degree sequence. But these properties are not enough to prove the isomorphism. Two graph with same degree sequence can be non-isomorphic. 

\begin{figure}[htb]
    \centering
    \begin{minipage}{0.45\textwidth}
        \centering
            \scalebox{.7}{ 
    \begin{tikzpicture}[auto, node distance=3cm, every loop/.style={},
                    thick,main node/.style={circle,draw,font=\sffamily\Large\bfseries}]

  \node[main node] (1) {1};
  \node[main node] (2) [below of=1] {2};
  \node[main node] (3) [below left of=2] {3};
  \node[main node] (4) [below  right of=2] {4};
  
  \node[main node] (5) [below  right of= 4] {5};
  \node[main node] (6) [below left of=4] {6};
  \node[main node] (7) [below left of=6] {7};

  \path[every node/.style={font=\sffamily\small}]
    (1) edge  node {} (2)
      
    (2) edge  node {} (3)
    edge  node {} (4)
    (3)  edge  node {} (4)

    (4) edge  node {} (5)
    edge  node {} (6)
    
    (5) 
    (6) edge  node {} (7)
    (7);
\end{tikzpicture}

}
\caption{Graph $G_1$} \label{fig:M1}

    \end{minipage}\hfill
    \begin{minipage}{0.45\textwidth}
        \centering
        \scalebox{.7}{ 
    
\begin{tikzpicture}[auto, node distance=3cm, every loop/.style={},
                    thick,main node/.style={circle,draw,font=\sffamily\Large\bfseries}]

  \node[main node] (1) {A};
  \node[main node] (2) [right of=1] {B};
  \node[main node] (3) [below of=1] {C};
  \node[main node] (4) [below  of=2] {D};
  
  \node[main node] (5) [below   of= 3] {E};
  \node[main node] (6) [below of=4] {F};
  \node[main node] (7) [below  of=5] {G};

  \path[every node/.style={font=\sffamily\small}]
    (1) edge  node {} (2)
      
    (2) 
    edge  node {} (4)
    (3)  edge  node {} (4)

    (4) edge  node {} (5)
    edge  node {} (6)
    
    (5) 
    (6) edge  node {} (7)
    edge  node {} (5)
    (7);
\end{tikzpicture}
}
\caption{Graph $G_2$}  \label{fig:M2}
    \end{minipage}
\end{figure}
\begin{figure}[htb]
\end{figure}
 Consider Fig. \ref{fig:M1} and Fig. \ref{fig:M2}. Brute-force  method is to try all the possibilities and to find the edge preserving bijection from vertex set of graph $G_1$ to vertex set of graph $G_2$. After trying all the possibilities one can find out that these two graphs are isomorphic. Using this brute-force method the isomorphism function $I: G_1  \rightarrow G_2$ is given in the Table \ref{tab:tab_F}. First, we will define the graphs that can be solved by the proposed method.
 Consider Fig. \ref{fig:M1} and Fig. \ref{fig:M2}. Brute-force  method is to try all the possibilities and to find the edge preserving bijection from vertex set of graph $G_1$ to vertex set of graph $G_2$. After trying all the possibilities one can find out that these two graphs are isomorphic. Using this brute-force method the isomorphism function $I: G_1  \rightarrow G_2$ is given in the Table \ref{tab:tab_F}. 

\textbf{Permissible graphs:}
 Let $G (  V,E)$\:be a simple, connected graph ,\\
\tab V : Set of vertices , $\left|V\right|$ = $n$ \\
\tab E : Set of edges (Ordered pair of vertices), \\
\tab $E \subset (V\times V)$, and for a given $v_i \in V $ and 
\\
$ S_i$  = $ { \{x \:|\: (x,v_i) \in E \}}$ , 
\: $\forall v_j , v_k \in S$,    $v_j \neq v_k  \implies dsv(v_j) \neq dsv(v_k)$ .\\
Where $dsv(p)$ is the neighbourhood degree list of a given vertex $p$ with the vertices.  For example, $dsv(4)=\{(5,1),(3,2),(6,2),(2,3)\}$. It is a list of tuples, where the first element is the vertex and second is its degree. The dsv is sorted by the degree (the second element).

\begin{figure}[htb]
    \centering
    \begin{minipage}{0.45\textwidth}
   
\captionof{table}{Isomorphism: $I(G_1,G_2)$} 

\begin{tabular}{|c|c|}

 \hline
 Vertex from graph $G_1$ & Vertex from graph $G_2$  \\ [0.7ex] 
\hline
 1 & G\\ 
 \hline
 2 & F\\ 
 \hline
 3 & E\\ 
 \hline
 4 & D\\ 
 \hline
 5 & C\\ 
 \hline
 6 & B\\ 
 \hline
 7 & A\\ 
 \hline
\end{tabular}
\label{tab:tab_F}

    \end{minipage}\hfill
    \begin{minipage}{0.45\textwidth}
        \centering
 \captionof{table}{{$dsv(i)$, for all $i \in V$ of the graph $G_{1}$ }}

     \begin{tabular}{|c||c|c|} 
\hline
{Vertex $i$}  &$degree(i)$ &$dsv(i)$ \\[0.7ex]
 
 \hline
\textbf{1} & 1&\{(2,3)\}  \\ 
\hline
\textbf{2} &3& \{(1,1),(3,2),(4,4)\}  \\ 
\hline
\textbf{3} &2& \{(2,3),(4,4)\}  \\ 
\hline
\textbf{4} &4& \{(5,1),(3,2),(6,2),(2,3)\}  \\ 
\hline
\textbf{5} &1& \{(4,4)\}  \\ 
\hline
\textbf{6} & 2&\{(7,1),(4,4)\}  \\ 
\hline
\textbf{7} & 1&\{(6,2)\}  \\ 
\hline

 \end{tabular}
    
    \label{tab:dsv1}

    \end{minipage}
\end{figure}
\par We will solve this problem using the proposed method. The input is two graphs $G_1$ and $G_2$ in their adjacency lists form, $list_1$,$list_2$ respectively. The algorithm will first check that whether these graphs follow the Definition of Permissible Graphs. or not. Table \ref{tab:list1} shows the adjacency list, $list_1$ of the graph $G_1$.
\begin{table}[htb]
  \caption{{$list_{1}$ of the graph $G_{1}$ }}
    \centering
     \begin{tabular}{|c||c|c|c|c|c|c|c|} 
\hline
\textbf{Vertex}  & list& list& list& list \\[0.7ex]
  & elements&elements&elements&elements\\[0.7ex]
 
 \hline
\textbf{1} & 2 & - & - & -  \\ 
 \hline
 \textbf{2} & 1 & 3 & 4 & -  \\ 
 \hline
\textbf{3} & 2 & 4 & - & -  \\ 
 \hline
\textbf{4} & 2 & 3 & 5 & 6  \\ 
 \hline
\textbf{5} & 4 & - & - & -  \\ 
 \hline
\textbf{6} & 4 & 7 & - & -  \\ 
 \hline
\textbf{7} & 6 & - & - & -  \\ 
 \hline

 \end{tabular}
    
    \label{tab:list1}
\end{table}
This algorithm consists of 3 sub modules:
\begin{enumerate}
\item  Preprocessing
\item Checking the input
\item Generate UIDs
\item Map Isomorphism
\end{enumerate}
\subsection{Preprocessing}
We will process on both the graphs and store the results for later use. Here $deg\_seq_i$ is list of degree sequences of the neighbours of the vertex $i$. This step will be performed on every vertex $i$, $i \in V $ for both the graphs. For understanding, $deg\_seq_4$ is shown in the Table \ref{tab:ds1}.    
\begin{algorithm}[htb] \caption{Preprocessiong}
\begin{algorithmic}[1]
\State ${deg\_seq_i} \gets 0$\;
\State $S$  = $ { \{x \:|\: (x,i) \in E \}}$
\ForEach {$v \in S$}
\State ${deg\_seq_i.append(dsv(v))}$ 
\EndFor
\end{algorithmic}
\end{algorithm}

\begin{table}[htb]
  \caption{degree sequence list, $deg\_seq_4$ of the graph $G_1$}
  
    \centering
 \begin{tabular}{|c |c|} 
 \hline
 vertex v&$dsv(v)$ of the neighbours v of the vertex 4 \\ [0.7ex] 
 \hline
 2&\{(1,1),(3,2),(4,4)\} \\ 
 \hline
 3&\{(2,3),(4,4)\} \\ 
 \hline
 6&\{(7,1),(4,4)\} \\ 
 \hline
 5&\{(4,4)\} \\ 
 \hline
\end{tabular}
    \label{tab:ds1}
\end{table}

\subsection{Checking the input}
In this module the algorithm will check whether the given graphs fits in the above definition or not. 
\begin{algorithm}\caption{Whether this method is applicable or not}
\begin{algorithmic}[1]
\State Sort the lists in $deg\_seq_i$ by the length of the lists and then lexicographically.
\State Sort $dsv(i)$ according to the order of lists in the $deg\_seq_i$.
\State $result \gets \textsf{true}$\;
\ForEach {$i \in V$}
\ForEach {$x,y \in deg\_seq_i$}
\If  {x=y}
\State $result \gets \textsf{false}$\;
\EndIf
\EndFor    
\EndFor    
\end{algorithmic}
\end{algorithm}
After getting the $deg\_seq_i$, for all the $i\in V$, the algorithm sorts the elements of this list. In the step 1. of the Algorithm 2, lists in the $deg\_seq_i$ are sorted by lexicographically (according to the second element in the tuple) and then by the length of the list. Now, we will perform this step on Vertices of the graph $G_1$.  
\begin{table*}[htb]
    \caption{sorted degree sequence list, $deg\_seq_i$ for $i \in V$of the graph $G_1$}
    \centering
 \begin{tabular}{|c |c|c|c|} 
 \hline
 vertex i&$deg\_seq_i$ & changed&remarks\\ [0.7ex] 
& &order of $dsv(i)$&\\ [0.7ex] 
 \hline
  1    & [ \{(1,1),(3,2),(4,4)\}] & \{(2,3)\}  &-\\ 
 \hline
  2    & [ \{(2,3)\}, \{(2,3),(4,4)\}, \{(5,1),(3,2),(6,2),(2,3)\} ] & \{(1,1),(3,2),(4,4)\}  &-\\ 
 \hline
  3    & [ \{(1,1),(3,2),(4,4)\}, \{(5,1),(3,2),(6,2),(2,3)\}] & \{(2,3),(4,4)\}  &-\\ 
 \hline
  4    & [ \{(4,4)\}, \{(2,3),(4,4)\}, \{(7,1),(4,4)\}, \{(1,1),(3,2),(4,4)\}] & \{(5,1),(6,2),(3,2),(2,3)\}  & changed\\ 
 \hline
  5    & [ \{(5,1),(3,2),(6,2),(2,3)\}] & \{(4,4)\}  &-\\ 
 \hline
  6    & [ \{(6,2)\}, \{(5,1),(3,2),(6,2),(2,3)\}] & \{(7,1),(4,4)\}  &-\\ 
 \hline
  7    & [ \{(7,1),(4,4)\}] & \{(6,2)\}  &-\\ 
 \hline

\end{tabular}
\label{tab:ds2}
\end{table*}

If value of the result is true then we can apply our algorithm. And it can go thorough out next stage.
\subsection{Generate UIDs}
Here we will introduce an encoding technique, in which we will assign different ids to all the vertices. And after that we will compare this ids in both the graphs among all the vertices. These unique ids (UID) are independent of the order in which we calculated these UIDs. 
\begin{algorithm}\caption{generate\_UID (count, i)}
\begin{algorithmic}[1]
\State tmp1=$dsv(i)[0]$\;
\While {$counter> 0$}
\State tmp2=[]
\ForEach{$x \in tmp1$}
\If  {$bool[x]\:=\:0$}
\State      $UID.append(dsv(x))$
\State      $tmp2.append(dsv(x)[0])$
\State    $count \gets count\: - \:1$\;
\Else
\State      $UID.append((x,-1))$
\EndIf
\EndFor
\State      $UID.append((-2,-2))$
\EndWhile    
\end{algorithmic}
\end{algorithm}
THe bool array is of size $|V|$ and it stores whether all the vertex are included or not. The array is initialized to 0 for all the elements. The count variable is initialized with total number of vertices, $|V|$. It will decrement its value by one as the algorithm explores the dsv of all the vertices. Function $generate\_UID$ is called for each vertex of graphs G1 and G2. This function takes count and the vertex i as its input. Bool array element will change its value to one, if that vertex is visited. The function will be called for every vertex $i$ with initial condition set to generate\_UID ($n$, $i$).  
Now, we will perform this operation on the vertex 4 of Graph $G_1$. 
\begin{table*}
    \caption{UID of vertex 4 of Graph $G_1$}
    \centering
 \begin{tabular}{|c||c|c|c|c|c|c|c|c|c|c|c|c|c|c|c|c|c|c|c|c|c|c|c|c|c|c|c|c|c|c|c|c|c|c|c|c|c|c|} 
 \hline
\textbf{vertex}&4&-2&5&6&3&2&-2&4 &7 &4 &2 &4 & 1& 3& 4&-2 &4 &6 &4 &2 &4 &2&3&4&-2  \\
\hline
\textbf{degree}&4&-2&1&2&2&3&-2&4 &1 &4 &3 & 4&1 &2 &4 &-2 &-1 &2 &-1 &-1 &-1 &3&-1&-1&-2 \\ 
 \hline

\end{tabular}
\label{tab:UID1}
\end{table*}

Now, we know from Table \ref{tab:tab_F} that vertex D from graph $G_2$ is isomorphic to the vertex 4 from graph $G_1$.
\begin{table*}
    \caption{UID of vertex D of Graph $G_2$}
    \centering
 \begin{tabular}{|c||c|c|c|c|c|c|c|c|c|c|c|c|c|c|c|c|c|c|c|c|c|c|c|c|c|c|c|c|c|c|c|c|c|c|c|c|c|c|} 
 \hline
\textbf{vertex}&D&-2&C&B&E&F&-2&D &A &D &F &D & G& E& D&-2 &D &6 &D &F &D &F&E&D&-2  \\
\hline
\textbf{degree}&D&-2&1&2&2&3&-2&4 &1 &4 &3 & 4&1 &2 &4 &-2 &-1 &2 &-1 &-1 &-1 &3&-1&-1&-2 \\ 
 \hline

\end{tabular}
\label{tab:UID2}
\end{table*}

\subsection{Map Isomorphism}
After generating UIDs from both the graph $G_1$ and $G_2$ the algorithm will try to map the UIDs of graph $G_1$ to UIDs of Graph $G_2$ and will report the isomorphism I. Here, $UID^{p}_j$ is UID of graph $p$ and of $j^{th}$ vertex. $UID^{p}$ is the list of all the UIDs of all the vertex of the graph $p$. The map Isomorphism function will be called for UID of every vertex of the graph $G_1$. The $iso$ is and array which maps vertex of one graph to the vertex of the other graph. If for all the UIDs the value of flag is 1 and all the $iso$ arrays give the same result then the two graphs are isomorphic and the isomorphism function is given by $iso$. The algorithm uses the constraint given in the Definition of Permissible Graphs., which is why the algorithm revolves around the degree sequence of the vertices.

\begin{algorithm}
\caption{map Isomorphism($UID^{1}_i$)}
\begin{algorithmic}[1]
\ForEach{$UID^{2}_j \in UID^{2}$}
\State      $iso\gets -1$
\State      $ind\gets 0$
\State      $flag\gets 0$
\If {$len(UID^{1}_{i}) =len(UID^{2}_{j})$ }
\While{$ind<len(UID^{1}_{i})$}
\State      $ind \gets ind + 1$
\If  {$iso[UID^{1}_{i}[ind][0]]\neq -1$ and $iso[UID^{1}_{i}[ind][0]]  \neq UID^{2}_{j}[ind][0]$}
\State      $flag\gets0$
\State      break
\EndIf
\If  {$UID^{1}_{i}[ind][1]\:\neq \:UID^{2}_{j}[ind][1]$}
\State      $flag\gets0$
\State      break
\Else 
\State $iso[UID^{1}_{i}[ind][0]]   \gets UID^{2}_{j}[ind][0]$
\State $flag\gets1$
\EndIf
\EndWhile
\If {flag=1}
\State break
\EndIf
\EndIf
\EndFor
\end{algorithmic}
\end{algorithm}

% \begin{restatable}{theorem}{thone}
% \label{theorem1}
% Adjacency matrix of a graph can be constructed by UID of any of its vertex.
% \end{restatable}
% \begin{proof}
% add proof
% \end{proof}
\begin{restatable}{theorem}{thtwo}
\label{theorem2}
Any vertex of a permissible graphs can not have more than one UID.
\end{restatable}
\begin{proof}
The UID consist of $UID_d$ (degrees) and $UID_n$ (nodes). 
\begin{itemize}
\item $UID_n[1]=v$ and $UID_d[1] = deg(v)$  .
\item $UID_n[2]=-2$ and $UID_d[2] = -2$
\item The next $deg(v)$ elements will be of $dsn(v)$.   

\item Between any two $i$ and $j$ such that, $UID_n[i]=-2$, $UID_d[i]=-2$ and $UID_n[j]=-2$, $UID_d[j]=-2$ and there is no k, where $i<k<j$, such that $UID_n[k]=-2$, $UID_d[k]=-2$:

\end{itemize}
\end{proof}

\begin{restatable}{theorem}{ththree}
\label{theorem3}
The graphs described in Definition of Permissible Graphs contains non-planar graphs. 
\end{restatable}
\begin{proof}
We will prove it by example.
\begin{figure}[htb]
\centering
    \scalebox{.7}{ 
    \begin{tikzpicture}[auto, node distance=3cm, every loop/.style={},
                    thick,main node/.style={circle,draw,font=\sffamily\Large\bfseries}]

  \node[main node] (1) {1};
  \node[main node] (2) [right of=1]{2};

  \node[main node] (3) [right of=2]{3};
  \node[main node] (4) [below of=1]{4};

  \node[main node] (5) [right of=4]{5};
  \node[main node] (6) [right of=5]{6};

  \node[main node] (7) [below of=5]{7};
  \node[main node] (8) [right of=7]{8};

  \path[every node/.style={font=\sffamily\small}]

    (1) edge  node {} (2)
    edge  node {} (6)
    edge  node {} (8)
    edge  node {} (5)
    edge  node {} (7)
    edge  node {} (3)
      
    (2) edge  node {} (1)
    edge  node {} (8)
    edge  node {} (5)
    edge  node {} (7)
    
    (3)  edge  node {} (8)
    edge  node {} (1)
    edge  node {} (5)
    edge  node {} (4)
    edge  node {} (7)

    (4) edge  node {} (7)
    edge  node {} (5)    
    edge  node {} (3)
    (5)     edge  node {} (1)
         edge  node {} (2)
         edge  node {} (3)
         edge  node {} (4)
         edge  node {} (6)
         edge  node {} (7)
         edge  node {} (8)
    
    (6) edge  node {} (5)
         edge  node {} (1)
         edge  node {} (7)

    (7)  edge  node {} (1)
         edge  node {} (2)
         edge  node {} (3)
         edge  node {} (4)
         edge  node {} (5)
         edge  node {} (6)
    (8)
         edge  node {} (2)
         edge  node {} (1)
         edge  node {} (5)
         edge  node {} (3)
    
    ;
\end{tikzpicture}

}
\caption{Graph $G_3$} \label{fig:np}

\end{figure}
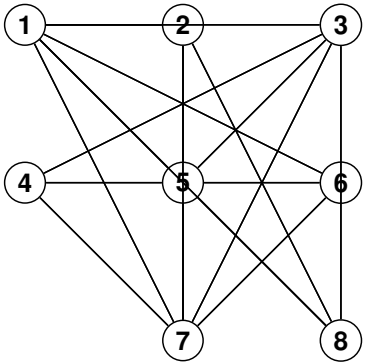
Consider Fig. \ref{fig:np}, number of vertices is 8 and number of edges is 19. The graph has a triangle in it. According to Euler's formula \cite{west2001introduction} this graph is non-planar graph ($3v-6 <= e$). Once can easily verify that this graph can also be solved using the proposed method.

\end{proof}

\section{Time complexity and practical results}
\label{sec:complexity}
\begin{enumerate}
\item The Prerprocessing step : 
Degree of each vertex can be found in $O(n^2)$ and after that their degree sequence can be calculated in $O(n^3)$. Overall, this step takes $O(n^3)$ time.
\item   Checking the input: 
Time omplexity $O(n^4)$.
\item Generate UIDs : 
The maximum length of any UID is of the order of the number of $O(n^{2})$. This function will be called n times, making this step run in $O(n^{3})$ time.

\item Try to fit Isomorphism:
For finding the isomorphism, the Algorithm 4 will compare UID of one graph with all the UIDs of the other graph. This is total $n \times n^{2}$ comparisons. This is $O(n^{3})$. The algorithm will be called $n$ times. So, the overall complexity becomes $n \times n^{3}$ making the complexity $O(n^{4})$. Overall, the algorithm runs in $O(n^4)$ time.
\end{enumerate}
The code was written in python and c++. In the Table \ref{tab:tab_pr} the number of permissible graphs for number of vertices ranging from 1 to 9 is given. We generated these graphs using SAGE \cite{sagemath} and nauty \cite{mckay1981practical}. For $n=9$, $10.66\%$ of the graphs lie under the Definition of Permissible Graphs. There are only $0.0171\%$ of graphs are tree for $n=9$. The Theorem \ref{theorem3} states that there are non-planar graphs, which can also be solved by the proposed method. 
\begin{table}
    \caption{Practical results}

    \centering
     \begin{tabular}{|c|c|c|c|c|c|c|} 
\hline
 Number  & Total& Connected&Trees&Planar  & Permissible& Fraction\\[0.7ex]
 of vertices & graphs& graphs& &graphs&  graphs & \\[0.7ex]
 
 \hline
1 & 1 & 1&1 &1& 1 &1 (trivial) \\ 
 \hline
2 & 2 & 1&1&2 & 1&0.5 \\ 
 \hline
3 & 4 & 2&1&4&0 &0\\ 
 \hline
4 & 11 & 6&2  &11  &1 &0.0909 \\ 
 \hline
5 & 34 &21&3& 33 &0&0 \\ 
 \hline
6 & 156 & 112&6&142 &6&0.038 \\ 
 \hline
7 & 1044 & 853&11 &822& 62 &0.05938 \\ 
 \hline
8 & 12346 &11117&23 & 6966& 1024 & 0.0829
 \\ 
 \hline
9 & 274668 & 261080&47& 79853&    29285&0.1066\\
 \hline
 \end{tabular}
    \label{tab:tab_pr}
\end{table}

\section{Conclusion}
The proposed algorithm solves GI problem for a special case. Before running any slower algorithm it is better to check the graphs using this algorithm. The time complexity is $O(n^{4})$, but these bounds are loose bounds. Because generating UID and matching the UIDs can be done faster, if we have used fast methods to sort them using the property that every value is bounded by $v$, making the bounds tighter. More importantly, this is polynomial time solution for this type of graphs.This algorithm is embarrassingly parallel making it useful in practical purposes.
\section*{Acknowledgment}
The authors would like to thank Dhruv Patel for helping in writing the software in python.
\medskip
\nocite{*}
\bibliographystyle{unsrt}%Used BibTeX style is unsrt
%\bibliography{sample.bbl}

\begin{thebibliography}{10}

\bibitem{babai2016graph}
L{\'a}szl{\'o} Babai.
\newblock Graph isomorphism in quasipolynomial time [extended abstract].
\newblock In {\em Proceedings of the 48th Annual ACM SIGACT Symposium on Theory
  of Computing}, pages 684--697. ACM, 2016.

\bibitem{barrus2015neighborhood}
Michael~D Barrus and Elizabeth Donovan.
\newblock Neighborhood degree lists of graphs.
\newblock {\em arXiv preprint arXiv:1507.08212}, 2015.

\bibitem{kelly1957congruence}
Paul~J Kelly et~al.
\newblock A congruence theorem for trees.
\newblock {\em Pacific J. Math}, 7(1):961--968, 1957.

\bibitem{aho1974design}
Alfred~V Aho and John~E Hopcroft.
\newblock {\em The design and analysis of computer algorithms}.
\newblock Pearson Education India, 1974.

\bibitem{hopcroft1974linear}
John~E Hopcroft and Jin-Kue Wong.
\newblock Linear time algorithm for isomorphism of planar graphs (preliminary
  report).
\newblock In {\em Proceedings of the sixth annual ACM symposium on Theory of
  computing}, pages 172--184. ACM, 1974.

\bibitem{lueker1979linear}
George~S Lueker and Kellogg~S Booth.
\newblock A linear time algorithm for deciding interval graph isomorphism.
\newblock {\em Journal of the ACM (JACM)}, 26(2):183--195, 1979.

\bibitem{colbourn1981testing}
Charles~J Colbourn.
\newblock On testing isomorphism of permutation graphs.
\newblock {\em Networks}, 11(1):13--21, 1981.

\bibitem{bodlaender1990polynomial}
Hans~L Bodlaender.
\newblock Polynomial algorithms for graph isomorphism and chromatic index on
  partial k-trees.
\newblock {\em Journal of Algorithms}, 11(4):631--643, 1990.

\bibitem{filotti1980polynomial}
Ion~Stefan Filotti and Jack~N Mayer.
\newblock A polynomial-time algorithm for determining the isomorphism of graphs
  of fixed genus.
\newblock In {\em Proceedings of the twelfth annual ACM symposium on Theory of
  computing}, pages 236--243. ACM, 1980.

\bibitem{babai1982isomorphism}
L{\'a}szl{\'o} Babai, D~Yu Grigoryev, and David~M Mount.
\newblock Isomorphism of graphs with bounded eigenvalue multiplicity.
\newblock In {\em Proceedings of the fourteenth annual ACM symposium on Theory
  of computing}, pages 310--324. ACM, 1982.

\bibitem{miller1983isomorphism}
Gary~L Miller.
\newblock Isomorphism testing and canonical forms for k-contractable graphs (a
  generalization of bounded valence and bounded genus).
\newblock In {\em International Conference on Fundamentals of Computation
  Theory}, pages 310--327. Springer, 1983.

\bibitem{luks1986parallel}
Eugene~M Luks.
\newblock Parallel algorithms for permutation groups and graph isomorphism.
\newblock In {\em Foundations of Computer Science, 1986., 27th Annual Symposium
  on}, pages 292--302. IEEE, 1986.

\bibitem{zager2008graph}
Laura~A Zager and George~C Verghese.
\newblock Graph similarity scoring and matching.
\newblock {\em Applied mathematics letters}, 21(1):86--94, 2008.

\bibitem{west2001introduction}
Douglas~Brent West et~al.
\newblock {\em Introduction to graph theory}, volume~2.
\newblock Prentice hall Upper Saddle River, 2001.

\bibitem{sagemath}
The~Sage Developers.
\newblock {\em {S}ageMath, the {S}age {M}athematics {S}oftware {S}ystem
  ({V}ersion x.y.z)}, YYYY.
\newblock {\tt http://www.sagemath.org}.

\bibitem{mckay1981practical}
Brendan~D McKay et~al.
\newblock Practical graph isomorphism.
\newblock 1981.

\bibitem{mckay2014practical}
Brendan~D McKay and Adolfo Piperno.
\newblock Practical graph isomorphism, ii.
\newblock {\em Journal of Symbolic Computation}, 60:94--112, 2014.

\bibitem{ullmann1976algorithm}
Julian~R Ullmann.
\newblock An algorithm for subgraph isomorphism.
\newblock {\em Journal of the ACM (JACM)}, 23(1):31--42, 1976.

\bibitem{foggia2001performance}
Pasquale Foggia, Carlo Sansone, and Mario Vento.
\newblock A performance comparison of five algorithms for graph isomorphism.
\newblock In {\em Proceedings of the 3rd IAPR TC-15 Workshop on Graph-based
  Representations in Pattern Recognition}, pages 188--199, 2001.

\bibitem{corneil1970efficient}
Derek~G Corneil and Calvin~C Gotlieb.
\newblock An efficient algorithm for graph isomorphism.
\newblock {\em Journal of the ACM (JACM)}, 17(1):51--64, 1970.

\bibitem{fortin1996graph}
Scott Fortin.
\newblock The graph isomorphism problem.
\newblock 1996.

\bibitem{luks1982isomorphism}
Eugene~M Luks.
\newblock Isomorphism of graphs of bounded valence can be tested in polynomial
  time.
\newblock {\em Journal of computer and system sciences}, 25(1):42--65, 1982.

\end{thebibliography}

\end{document}